\documentclass[a4paper,UKenglish]{lipics-v2021}

\usepackage{todonotes}
\usepackage[fixamsmath,disallowspaces]{mathtools}
\usepackage{fixmath}
\usepackage{bm}
\usepackage{amsmath}
\usepackage{todonotes}
\usepackage{xspace}
\usepackage{booktabs}
\newcommand{\complClFont}[1]{\normalfont\textbf{#1}}         
\newcommand{\logicClFont}[1]{\mathcal{#1}}        
\newcommand{\problemFont}[1]{\mathrm{#1}}         

\usetikzlibrary{arrows}

\usepackage{mathtools}
\DeclarePairedDelimiterX\set[1]\lbrace\rbrace{#1}

\newcommand{\splits}{\protect\ensuremath{\#\mathsf{splits}}}
\newcommand{\teamsize}
	{\protect\ensuremath{|T|}}
\newcommand{\tw}{\protect\ensuremath{\mathsf{tw}}}

\newcommand{\twstruc}{\protect\ensuremath{\tw(\calA)}}
\newcommand{\freevar}{\protect\ensuremath{\#\mathsf{free\text-variables}}}
\newcommand{\quantifier}{\protect\ensuremath{\#\forall}}
\newcommand{\strucsize}{\protect\ensuremath{{|\calA|}}}
\newcommand{\formula}{\protect\ensuremath{|\Phi|}}
\newcommand{\arity}
	{\protect\ensuremath{\mathsf{dep\text-arity}}}
\newcommand{\variable}
	{\protect\ensuremath{\#\mathsf{variables}}}

\newcommand{\depa}[2]{{\mathsf{dep}}({\mathbf{#1}};{\mathbf{#2}})}

\newcommand{\depas}[2]{{\mathsf{dep}}({#1};{#2})}

\newcommand{\para}{\protect\ensuremath{\complClFont{para}}} 
\newcommand{\FPT}{\protect\ensuremath{\complClFont{FPT}}} 
\newcommand{\NP}{\protect\ensuremath{\complClFont{NP}}} 
\newcommand{\PSPACE}{\protect\ensuremath{\complClFont{PSPACE}}} 
\newcommand{\Ptime}{\protect\ensuremath{\complClFont{P}}} 
\newcommand{\XP}{\protect\ensuremath{\complClFont{XP}}}
\newcommand{\WP}{\protect\ensuremath{\complClFont{W[P]}}} 
\newcommand{\XNP}{\protect\ensuremath{\complClFont{XNP}}}
\newcommand{\NEXP}{\protect\ensuremath{\complClFont{NEXP}}}

\newcommand{\hard}{\text{-h}}
\newcommand{\complete}{\text{-c}}

\newcommand{\calA}{\mathcal{A}}

\newcommand{\data}{\mathsf{dc}}
\newcommand{\expression}{\mathsf{ec}}
\newcommand{\combined}{\mathsf{cc}}

\newcommand{\paraMC}[2]{{#2}\text-#1}

\newcommand{\preduction}{\ensuremath{\leq^{\Ptime}_m}}
\newcommand{\fptreduction}{\ensuremath{\leq^{\FPT}}}

\newcommand{\dfn}{\mathrel{\mathop:}=}

\newcommand{\setdefination}[1]{\protect\ensuremath{\{\, #1 \,\} }} 

\newcommand{\FO}{\protect\ensuremath{\logicClFont{FO}}\xspace}
\newcommand{\PDL}{\protect\ensuremath{\logicClFont{PDL}}\xspace}
\newcommand{\PL}{\protect\ensuremath{\logicClFont{PL}}\xspace}
\newcommand{\D}{\protect\ensuremath{\logicClFont{D}}\xspace}

\newcommand{\ESO}{\protect\ensuremath{\logicClFont{ESO}}\xspace}

\newcommand{\MC}{\problemFont{MC}}
\newcommand{\SAT}{\problemFont{SAT}}

\newcommand{\problemdef}[3]{%
\begin{center}
\begin{tabular}{lp{9cm}}\toprule
\textsf{\bfseries Problem:}& #1 \\\midrule
\textsf{\bfseries Input:}& #2.\\
\textsf{\bfseries Question:}& #3?\\\bottomrule
\end{tabular}
\end{center}
}

\newcommand{\VAR}{\mathrm{VAR}}
\newcommand{\Fr}{\mathrm{Fr}}
\newcommand{\dom}{\mathrm{dom}}

\usepackage{etoolbox}

\usetikzlibrary{decorations.pathreplacing}
\makeatletter
\newcounter{aaa}
\tikzset{
  apply/.style args={#1 except on segments #2}{postaction={
      /utils/exec={
        \@for\mattempa:=#2\do{\csdef{aaa@\mattempa}{}}
        \setcounter{aaa}{0}
      },
      decorate,decoration={show path construction,
        moveto code={},
        lineto code={
          \stepcounter{aaa}
          \ifcsdef{aaa@\theaaa}{}{
            \path[#1] (\tikzinputsegmentfirst) -- (\tikzinputsegmentlast);
          }
        },
        curveto code={
          \stepcounter{aaa}
          \ifcsdef{aaa@\theaaa}{}{
            \path [#1] (\tikzinputsegmentfirst) .. controls
            (\tikzinputsegmentsupporta) and (\tikzinputsegmentsupportb)
            ..(\tikzinputsegmentlast);
          }
        },
        closepath code={
          \stepcounter{aaa}
          \ifcsdef{aaa@\theaaa}{}{
            \path [#1] (\tikzinputsegmentfirst) -- (\tikzinputsegmentlast);
          }
        },
      },
    },
  },
}
\makeatother

\definecolor{paraNEXP}{HTML}{0F2080}
\definecolor{paraPSPACE}{HTML}{A95AA1}
\definecolor{paraNP}{HTML}{F5793A}
\definecolor{FPT}{HTML}{85C0F9}

\renewcommand{\mathbf}[1]{\boldsymbol{#1}}

\hideLIPIcs  

\title{A Parameterized View on the Complexity of Dependence Logic}

\author{Juha Kontinen}{University of Helsinki, Department of Mathematics and Statistics, Helsinki, Finland}{email}{https://orcid.org/0000-0003-0115-5154}{Funded by grants 308712 and 338259 of the Academy of Finland}

\author{Arne Meier}{Leibniz Universität Hannover, Institut für Theoretische Informatik, Hannover, Germany}{meier@thi.uni-hannover.de}{https://orcid.org/0000-0002-8061-5376}{Funded by the German Research Foundation (DFG), project ME4279/1-2}

\author{Yasir Mahmood}{Leibniz Universität Hannover, Institut für Theoretische Informatik, Hannover, Germany}{mahmood@thi.uni-hannover.de}{https://orcid.org/0000-0002-5651-5391}{Funded by the German Research Foundation (DFG), project ME4279/1-2}

\authorrunning{J.\ Kontinen, A.\ Meier, and Y.\ Mahmood} 

\Copyright{Juha Kotinen, Arne Meier, and Yasir Mahmood}



\nolinenumbers 

\begin{document}

\maketitle

\begin{abstract}
In this paper, we investigate the parameterized complexity of model checking for Dependence Logic which is a well studied logic in the area of Team Semantics. 
We start with a list of nine immediate parameterizations for this problem, namely: the number of disjunctions (i.e., splits)/(free) variables/universal quantifiers, formula-size, the tree-width of the Gaifman graph of the input structure, the size of the universe/team, and the arity of dependence atoms.
We present a comprehensive picture of the parameterized complexity of model checking and obtain a division of the problem into tractable and various intractable degrees.
Furthermore, we also consider the complexity of the most important  variants (data and expression complexity) of  the model checking problem by fixing parts of the input. 
\end{abstract}

\section{Introduction}


In this article, we explore the parameterized complexity of model checking for dependence logic ($\D$). 
We give a concise classification of this problem and its standard variants (expression and data complexity) with respect to several syntactic and structural parameters. 
Our results lay down a solid foundation for a systematic study of the parameterized complexity of team-based logics.

The introduction of Dependence Logic \cite{DBLP:books/daglib/0030191} in 2007 marks also the birth of the general semantic framework of team semantics that has enabled a systematic study of various notions of dependence and independence during the past decade. 
Team semantics differs from Tarski's semantics by interpreting formulas by sets of assignments instead of a single assignment as in first-order logic. 
Syntactically, dependence logic is an extension of first-order logic by new dependence atoms $\depa{x}{y}$ expressing that the values of variables $\mathbf x$ functionally determine  values of the variables $\mathbf y$ (in the team under consideration). 
Soon after the introduction of dependence logic many other interesting team-based logics and atoms were introduced such as \emph{inclusion}, \emph{exclusion}, and \emph{independence} atoms that are intimately connected to the corresponding inclusion, exclusion, and multivalued dependencies studied in database theory~\cite{gradel10,galliani12}. 
Furthermore, the area has expanded, e.g., to propositional, modal and probabilistic variants (see \cite{HannulaKVV18,Luck19,HannulaKBV20} and the references therein).

For the applications, it is important to understand the complexity theoretic aspects of dependence logic and its variants. 
In fact, during the past few years, these aspects have been addressed in several studies. 
For example, on the level of sentences dependence logic and independence logic are equivalent to existential second-order logic while inclusion logic corresponds to positive greatest fixed point logic and thereby captures $\Ptime$  over finite (ordered) structures \cite{gallhella13}. 
Furthermore, there are (non-parameterized) studies that restrict the syntax and try to pin the intractability of a problem to a particular (set of) connective(s).
For instance, Durand and Kontinen~\cite{durand11} characterize the data complexity of fragments of dependence logic with bounded arity of dependence atoms/number of universal quantifiers, and Grädel~\cite{Gradel13} characterizes the combined and the expression complexity of the model checking problem of dependence logic. 
These studies will be of great help in developing  our parameterized approach.

A formalism to enhance the understanding of the inherent intractability of computational problems is brought by the framework of parameterized complexity~\cite{DBLP:series/txcs/DowneyF13}. 
Initiated by the founding fathers Downey and Fellows, in this area within computational complexity theory one strives for more structure within the darkness of intractability. 
Essentially, one tries to identify so-called parameters of a considered problem $\Pi$ to find algorithms solving $\Pi$ with runtimes of the form $f(k)\cdot|x|^{O(1)}$ for inputs $x$, corresponding parameter values $k$, and a computable function $f$. 
These kind of runtimes are called \emph{$\FPT$-runtimes} (from fixed-parameter tractable; short $\FPT$) and tame the combinatoric explosion of the solution space to a function $f$ in the parameter. 
As a very basic example in this vein, we can consider the propositional satisfiability problem $\SAT$. 
An immediate parameter that pulls the problem into the class $\FPT$ is the number of variables, as one can solve  $\SAT$ in time $2^{k}\cdot|\varphi|$ if $k$ is the number of variables of a given propositional formula $\varphi$. 
Yet, this parameter is not very satisfactory as it neither is seen fixed nor slowly growing in its practical instances. 
However, there are several interesting other parameters under which $\SAT$ becomes fixed-parameter tractable, e.g., the so-called treewidth of the underlying graph representations of the considered formula~\cite{DBLP:series/faia/SamerS09}. 
This term was coined by Robertson and Seymour in 1984~\cite{DBLP:journals/jct/RobertsonS84} and established a profound position (currently DBLP lists 812 papers with treewidth in its title) also in the area of parameterized complexity in the last years~\cite{DBLP:conf/sofsem/Bodlaender05,DBLP:series/txcs/DowneyF13}. 

Coming back to fpt-runtimes, a runtime of a very different quality (yet still polynomial for fixed parameters) than $\FPT$ is summarized by the complexity class $\XP$: $|x|^{f(k)}$ for inputs $x$, corresponding parameter values $k$, and a computable function $f$. 
Furthermore, analogously as $\XP$ but on nondeterministic machines, the class $\XNP$ will be of interest in this paper.
Further up in the hierarchy, classes of the form $\para\mathcal{C}$ for a classical complexity class $\mathcal{C}\in\{\NP,\PSPACE,\NEXP\}$ play a role in this paper.
Such classes intuitively capture all problems that are in the complexity class $\mathcal C$ after fpt-time preprocessing. 
In Fig.~\ref{fig:cc-landscape} an overview of these classes and their relations are depicted (for further details see, e.g., the work of Elberfeld~et~al.~\cite{DBLP:journals/algorithmica/ElberfeldST15}).
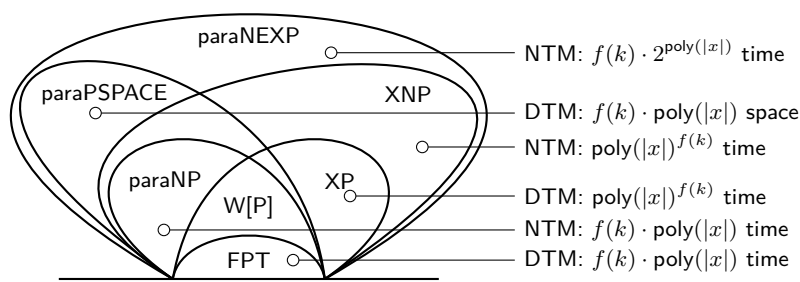
\begin{figure}[t]
	\centering
		\begin{tikzpicture}[every node/.style={font=\sffamily\footnotesize}]

		\path[use as bounding box] (-.7,0) rectangle (10,3.6);
	
		\draw[black,thick] (1.5,0) edge[out=150,in=30,looseness=12] (3.5,0);
		\node at (2.5,3.2) {paraNEXP};
		
		\draw[black,thick] (1.5,0) .. controls (-3.8,3.7) and (3,4) .. (3.5,0);
		\node at (.6,2.45) {paraPSPACE};

		\draw[black,thick] (3.5,0) .. controls (10.4,4.6) and (-3.1,2.9) .. (1.5,0);
		\node at (4.6,2.45) {XNP};

		\node at (2.5,.95) {W[P]};

		\draw[black,thick] (1.5,0) edge[out=145,in=100,looseness=4] (3.5,0);
		\node at (1.4,1.3) {paraNP};
		
		\draw[black,thick] (1.5,0) edge[out=80,in=35,looseness=4] (3.5,0);
		\node at (3.7,1.3) {XP};
	
		\draw[black,thick] (1.5,0) edge[out=80,in=100] (3.5,0);
		\node at (2.5,0.25) {FPT};
				
		\draw[black, thick, rounded corners] (0,0) -- (5,0);

\node[anchor=west] at (6,0.25)  (fptl){DTM: $f(k)\cdot \text{poly}(|x|)$ time};
\draw[o-]  (3,0.25)  -- (fptl) ;

\node[anchor=west] at (6,.65) (pnp){NTM: $f(k)\cdot \text{poly}(|x|)$ time};
\draw[o-]  (1.3,.65) -- (pnp) ;

\node[anchor=west] at (6,1.1)  (xp){DTM: $\text{poly}(|x|)^{f(k)}$ time};
\draw[o-]  (3.75,1.1)   -- (xp) ;

\node[anchor=west] at (6,2.2) (ppsp){DTM: $f(k)\cdot \text{poly}(|x|)$ space};
\draw[o-]  (0.4,2.2)   -- (ppsp) ;

\node[anchor=west] at (6,1.75) (xnp){NTM: $\text{poly}(|x|)^{f(k)}$ time};
\draw[o-]  (4.7,1.75)   -- (xnp) ;

\node[anchor=west] at (6,3) (nxp){NTM: $f(k)\cdot 2^{\text{poly}(|x|)}$ time};
\draw[o-]  (3.5,3)   -- (nxp) ;

	\end{tikzpicture}
	\caption{Landscape showing relations of relevant parameterized complexity classes with machine definitions.}\label{fig:cc-landscape}
\end{figure}

Recently, the propositional variant of dependence logic (\PDL) has been investigated regarding its parameterized complexity~\cite{DBLP:conf/foiks/MeierR18,DBLP:conf/foiks/0002M20}.
Moreover, propositional independence and inclusion logic have also been studied from the perspective of parameterized complexity~\cite{DBLP:journals/corr/abs-2105-14887}.
In this paper, we further pursue the parameterized journey through the world of team logics and will visit the problems of first-order dependence logic $\D$. 
As this paper is the first one that investigates $\D$ from the parameterized point of view, we need to gather the existing literature and revisit many results particularly from this  perspective. 
As a result, this paper can be seen as a systematic study with some  of the result following in a straightforward manner from the known non-parameterized results and some shedding light also on the non-parameterized view of model checking. 

We give an example below to illustrate how the concept of dependence arises as a natural phenomenon in the physical world.
\begin{example}
\begin{table}[t]
	\centering
	\begin{tabular}{cccccc}\toprule
\texttt{Flight} 	&  	\texttt{Destination} 	& 	\texttt{Gate} & \texttt{Date} & \texttt{Time} \\ \toprule

  FIN-70	& HEL -- FI	& C1 	& 04.10.2021		& 09:55  	\\
  SAS-475	& OSL -- NO	& C3 	& 04.10.2021		& 12:25 		 \\
  SAS-476	& HAJ -- DE	& C2 	& 04.10.2021		& 12:25 		 \\
  FIN-80	& HEL -- FI	& C1 	& 04.10.2021		& 19:55 	 	\\
  KLM-615	& ATL -- USA	& A5		& 05.10.2021 	& 11:55 		\\
  QR-70		& DOH -- QR	& B6 	& 05.10.2021 	& 12:25		\\
  THY-159	& IST -- TR	& A1 	& 05.10.2021		& 15:55 	 	\\
  FIN-80	& HEL -- FI	& C1 	& 05.10.2021		& 19:55 	 	\\
\bottomrule
	\end{tabular}
	\caption{An example flight departure screen at an airport}\label{database:example}
\end{table}	
	The database in Table~\ref{database:example} presents a  screen at an airport for showing details about departing flights.
	Alternatively, it can be seen as a team $T$ over attributes in the top row as variables.
	Clearly, $T\models \depas{\texttt{Flight,Date,Time}}{\texttt{Destination,Gate}}$, as well as $T\models \depas{\texttt{Gate,Date,Time}}{\texttt{Destination, Flight}}$. 
	Whereas, $T\not\models \depas{\texttt{Destination,Gate}}{\texttt{Time}}$ as witnessed by the pair (FIN-$70$, HEL -- FI, C$1$ , $04.10.2021$, $09:55$) and (FIN-$80$, HEL -- FI, C$1$ , $04.10.2021$, $19:55$).
\end{example}

\subparagraph*{Contribution.} Our classification is two-dimensional: 
\begin{enumerate}
	\item We consider the model checking problem of \D under various parameterizations: number of split-junctions in a formula $\splits$, the length of the formula $\formula$, number of free variables $\freevar$, the treewidth of the structure $\twstruc$, the size of the structure $\strucsize$, the size of the team $\teamsize$, the number of universal quantifiers in the formula $\quantifier$, the arity of the dependence atoms $\arity$, as well as the total number of variables $\variable$. 
	\item We distinguish between expression complexity $\expression$ (the input structure is fixed), data complexity $\data$ (the formula is fixed), and combined complexity $\combined$. 
\end{enumerate}
The results are summarized in Table~\ref{tbl:overview}. 
For instance, the parameters $\quantifier,\arity,$ and $\variable$ impact in lowering the complexity for $\expression$ (and not for $\combined$ or $\data$), while the parameter $\strucsize$ impacts for $\data$ but not for $\combined$ or $\expression$.

\subparagraph*{Related work.} The parameterized complexity analyses in the propositional setting~\cite{DBLP:conf/foiks/MeierR18,DBLP:conf/foiks/0002M20,DBLP:journals/corr/abs-2105-14887} have considered the combined complexity of model checking and satisfiability as problems of interest. 
On the $\combined$-level, the picture there is somewhat different, e.g., team size as a parameter for propositional dependence logic  enabled a $\FPT$ algorithm while in our setting it has no effect on the complexity  ($\para\NEXP$).
Grädel~\cite{Gradel13} studied the expression and the combined complexity for $\D$ in the classical setting, whereas the data complexity was considered by Kontinen~\cite{jarmoKontinen13}.

\subparagraph*{Organization of the paper.} 
In Section~\ref{sec:prelims}, we introduce the foundational concepts of dependence logic as well as parameterized complexity. 
In Section~\ref{sec:results} our results are presented while Section~\ref{sec:concl} concludes the article.

\section{Preliminaries}\label{sec:prelims}
We require standard notions from classical complexity theory \cite{DBLP:books/daglib/0072413}. 
We encounter the classical complexity classes $\Ptime, \NP, \PSPACE, \NEXP$ and their respective completeness notions, employing polynomial time many-one reductions ($\preduction $). 

\subparagraph*{Parameterized Complexity Theory.}
A \emph{parameterized problem} (PP) $P\subseteq\Sigma^*\times\mathbb N$ is a subset of the crossproduct of an alphabet and the natural numbers.
For an \emph{instance} $(x,k)\in\Sigma^*\times\mathbb N$, $k$ is called the (value of the) \emph{parameter}.
A \emph{parameterization} is a polynomial-time computable function that maps a value from $x\in\Sigma^*$ to its corresponding $k\in\mathbb N$.
The problem $P$ is said to be \emph{fixed-parameter tractable} (or in the class $\FPT$) if there exists a deterministic algorithm $\mathcal A$ and a computable function $f$ such that for all $(x,k)\in\Sigma^*\times \mathbb N$, algorithm $\mathcal A$ correctly decides the membership of $(x,k)\in P$ and runs in time $f(k)\cdot|x|^{O(1)}$.
The problem $P$ belongs to the class $\XP$ if $\mathcal A$ runs in time $|x|^{f(k)}$ on a deterministic machine, whereas $\XNP$ is the non-deterministic counterpart of $\XP$.
Abusing a little bit of notation, we write $\mathcal C$-machine for the type of machines that decide languages in the class $\mathcal C$, and we will say a function $f$ is ``$\mathcal C$-computable'' if it can be computed by a machine on which the resource bounds of the class $\mathcal C$ are imposed.

Also, we work with classes that can be defined via a precomputation on the parameter.
\begin{definition}
	Let $\mathcal C$ be any complexity class.
	Then $\para\mathcal C$ is the class of all PPs $P\subseteq\Sigma^*\times\mathbb N$ such that there exists a computable function $\pi\colon\mathbb N\to\Delta^*$ and a language $L\in\mathcal C$ with $L\subseteq\Sigma^*\times\Delta^*$ such that for all $(x,k)\in\Sigma^*\times\mathbb N$ we have that $(x,k)\in P \Leftrightarrow (x,\pi(k))\in L$.
\end{definition}
Notice that $\para\complClFont{P}=\FPT$.
The complexity classes $\mathcal{C}\in\{\NP,\PSPACE,\NEXP\}$ are used in the $\para\mathcal C$ context by us.

A problem $P$ is in the complexity class $\WP$, if it can be decided by a NTM running in time $f(k)\cdot|x|^{O(1)}$ steps, with at most $g(k)$-many non-deterministic steps, where $f,g$ are computable functions.
Moreover, $\WP$ is contained in the intersection of $\para\NP$ and $\XP$ (for details see the textbook of Flum and Grohe~\cite{DBLP:series/txtcs/FlumG06}). 

Let $c\in\mathbb N$ and $P\subseteq\Sigma^*\times\mathbb N$ be a PP, then the \emph{$c$-slice of $P$}, written as $P_c$ is defined as $P_c:=\{\,(x,k)\in\Sigma^*\times\mathbb N\mid k=c\,\}$.
Notice that $P_c$ is a classical problem then.
Observe that, regarding our studied complexity classes, showing membership of a PP $P$ in the complexity class $\para\mathcal C$, it suffices to show that for each slice $P_c\in\mathcal C$ is true.

\begin{definition}\label{def:fpt-reduction}
	Let $P\subseteq\Sigma^*\times\mathbb N,Q\subseteq\Gamma^*$ be two PPs.
	One says that $P$ is \emph{fpt-reducible} to $Q$, $P\fptreduction Q$, if there exists an $\FPT$-computable function $f\colon\Sigma^*\times\mathbb N\to\Gamma^*\times\mathbb N$ such that
	\begin{itemize}
		\item for all $(x,k)\in\Sigma^*\times\mathbb N$ we have that $(x,k)\in P\Leftrightarrow f(x,k)\in Q$,
		\item there exists a computable function $g\colon\mathbb N\to\mathbb N$ such that for all $(x,k)\in\Sigma^*\times\mathbb N$ and $f(x,k)=(x',k')$ we have that $k'\leq g(k)$.
	\end{itemize}
\end{definition}
Finally, in order to show that a problem $P$ is $\para\mathcal C$-hard (for some complexity class $\mathcal C$) it is enough to prove that for some $c\in \mathbb N$, the slice $P_c$ is $\mathcal C$-hard in the classical setting.
\subparagraph*{Dependence Logic.}
We assume basic familiarity with predicate logic~\cite{DBLP:books/daglib/0082516}. 
We consider first-order vocabularies $\tau$ that are sets of \emph{function} symbols and \emph{relation} symbols with an equality symbol $=$. 
Let $\VAR$ be a countably infinite set of \emph{first-order variables}.
Terms over $\tau$ are defined in the usual way, and the set of well-formed formulas of first order logic ($\FO$) is defined by the following BNF:
\[
	\psi \Coloneqq
	t_1 =t_2\mid 
	R(t_1,\dots,t_k)\mid
	\lnot R(t_1,\dots,t_k)\mid
	\psi\land\psi\mid
	\psi\lor\psi\mid
	\exists x\psi\mid
	\forall x\psi,
\]
where $t_i$ are terms $1\leq i\leq k$, $R$ is a $k$-ary relation symbol from $\sigma$, $k\in\mathbb N$, and $x\in\VAR$.
If $\psi$ is a formula, then we use $\VAR(\psi)$ for its set of variables, and $\Fr(\psi)$ for its set of free variables.
We evaluate $\FO$-formulas in $\tau$-structures, which are pairs of the form $\calA=(A,\tau^\calA)$, where $A$ is the \emph{domain} of $\calA$ (when clear from the context, we write $A$ instead of $\dom(\calA)$), and $\tau^\calA$ interprets the function and relational symbols in the usual way (e.g., $t^\calA\langle s\rangle=s(x)$ if $t=x\in\VAR$).
If $\mathbf t=(t_1,\dots,t_n)$ is a tuple of terms for $n\in\mathbb N$, then we write $\mathbf t^\calA\langle s\rangle$ for $(t_1^\calA\langle s\rangle, \dots, t_n^\calA\langle s\rangle)$.

Dependence logic ($\D$) extends $\FO$ by dependence atoms of the form $\depa{t}{u}$ where $\mathbf t$ and $\mathbf u$ are tuples of terms.
The semantics is defined through the concept of a team.
Let $\calA$ be a structure and $X\subseteq\VAR$, then an \emph{assignment} $s$ is a mapping $s\colon X\rightarrow A$. 
\begin{definition}
Let $X\subseteq\VAR$. A \emph{team $T$ in $\calA$ with domain $X$} is a set of assignments $s\colon X\to A$.	
\end{definition}
For a team $T$ with domain $X\supseteq Y$ define its \emph{restriction} to $Y$ as $T\upharpoonright Y\coloneqq\{\,s\upharpoonright Y \mid s\in T\,\}$.
If $s\colon X\to A$ is an assignment and $x\in\VAR$ is a variable, then $s^x_a\colon X\cup\{x\}\to A$ is the assignment that maps $x$ to $a$ and $y\in X\setminus\{x\}$ to $s(y)$. 
Let $T$ be a team in $\calA$ with domain $X$. 
Then we define $f\colon T\to \mathcal{P}(A)\setminus\{\emptyset\}$ as the \emph{supplementing function} of $T$.
This is used to extend or modify $T$ to the \emph{supplementing team} $T^x_f\coloneqq\{\,s^x_a\mid s\in T,a\in f(s)\,\}$. 
For the case $f(s)=A$ is the constant function we simply write $T^x_\calA$ for $T^x_f$. 
The semantics of $\D$-formulas is defined as follows.
\begin{definition}\label{def-semantics}
	Let $\tau$ be a vocabulary, $\calA$ be a $\tau$-structure and $T$ be a team over $\calA$ with domain $X\subseteq\VAR$. Then,
\begin{alignat*}{3}
	& (\calA,T)\models t_1=t_2 && \;\text{ iff }\; && \forall s\in T: t_1^\calA\langle s\rangle=t_2^\calA\langle s\rangle\\
	& (\calA,T)\models R(t_1,\ldots,t_n) && \;\text{ iff }\;  && \forall s\in T: (t_1^\calA\langle s\rangle,\ldots,t_n^\calA\langle s\rangle)\in R^{\calA}\\
	& (\calA,T)\models \neg R(t_1,\ldots,t_n) && \;\text{ iff }\;  && \forall s\in T:  (t_1^\calA\langle s\rangle,\ldots,t_n^\calA\langle s\rangle)\not\in R^{\calA}\\
	& (\calA,T)\models \depa{t}{u} && \;\text{ iff }\;  && \forall s_1, s_2\in T: \mathbf t^\calA\langle s_1\rangle=\mathbf t^\calA\langle s_2\rangle \implies  \mathbf u^\calA\langle s_1\rangle=\mathbf u^\calA\langle s_2\rangle \\		
	& (\calA,T)\models \phi_0\land \phi_1  && \;\text{ iff }\;  && (\calA,T)\models \phi_0 \quad \text{ and }\quad (\calA,T)\models \phi_1 \\
	& (\calA,T)\models \phi_0\lor \phi_1  && \;\text{ iff }\;  && \exists T_0\exists T_1: T_0\cup T_1=T \quad \text{ and } \quad (\calA,T_i)\models \phi_i  \, \text{ for }i =0,1\\
	& (\calA,T)\models\exists x\phi && \;\text{ iff }\;  && (\calA,T^x_f)\models\phi\text{ for some }f\colon T\to \mathcal{P}(A)\setminus\{\emptyset\}\\
	& (\calA,T)\models\forall x\phi && \;\text{ iff }\;  && (\calA,T^x_\calA)\models\phi
\end{alignat*}

\end{definition}

Notice that we only consider formulas in negation normal form (NNF) as any formula of dependence logic can be transformed into logically equivalent NNF-form. 
Further note that $(\calA,T)\models \phi$ for all $\phi$ when $T=\emptyset$ (this is also called the \emph{empty team property}).
Furthermore, $\D$-formulas are \emph{local}, that is, for a team $T$ in $\calA$ over domain $X$ and a $\D$-formula $\phi$, we have that $(\calA,T)\models\phi$ if and only if $(\calA,T\upharpoonright \Fr(\phi))\models\phi$. 
Finally, every $\D$-formula $\phi$,  if $(\calA,T)\models \phi$ then $(\calA,P)\models \phi$ for every $P\subseteq T$. 
This property is known as the downwards closure.

\begin{definition}[Gaifman graph]
	Given a vocabulary $\tau$ and a $\tau$-structure $\calA$, the {Gaifman graph} $G_{\calA}=(A,E)$ of $\calA$ is defined as
	\begin{align*}
		E\dfn\big\{\,\{u,v\}\;\big|\; &\text{ if there is an } R^n\in\tau \text{ and } \mathbf a\in A^n \text{ with } R^\calA(\mathbf a) \text{ and }u,v\in \mathbf a\,\big\}.
	\end{align*}
	That is, there is a relation $R \in \tau$ of arity $n$ such that $u$ and $v$ appear together in $R^{\calA}$.
\end{definition}

Intuitively, the Gaifman graph of a structure $\calA$ is an undirected graph with the universe of $\calA$ as vertices and connects two vertices when they share a tuple in a relation (see also Fig.~\ref{fig:tw-ex}). 

%
%
%

\begin{definition}[Treewidth]\label{def-tw}
The \emph{tree decomposition} of a given graph $G=(V,E)$ is a tree $T=(B,E_T)$, where the vertex set $B\subseteq\mathcal P(V)$ is the collection of \emph{bags} and $E_T$ is the edge relation such that the following is true.
\begin{itemize}
	\item $\bigcup_{b\in B}=V$,
	\item for every $\{u,v\}\in E$ there is a bag $b\in B$ with $u,v\in b$, and 
	\item for all $v\in V$ the restriction of $T$ to $v$ (the subset with all bags containing $v$) is connected.
\end{itemize}
The \emph{width} of a given tree decomposition $T=(B,E_T)$ is the size of the largest bag minus one: $\max_{b\in B}|b|-1$.
The \emph{treewidth} of a given graph $G$ is the minimum over all widths of tree decompositions of $G$.
\end{definition}
Observe that if $G$ is a tree then the treewidth of $G$ is one.
Intuitively, one can say that treewidth accordingly is a measure of tree-likeness of a given graph. 
\begin{example}\label{tw-ex}
\begin{figure}[t]
\begin{tikzpicture}

\node (table) at (-2,0) {\begin{tabular}{cccccc}\toprule
\texttt{Flight}  & 	\texttt{Gate} & \texttt{Time} \\ \toprule
  \textbf{F}IN-\textbf{7}0	&  C1 	& 0\textbf{9}:55  	\\
  \textbf{S}AS-47\textbf{5}	&  C3 	& \textbf{12}:25 		 \\
  \textbf{S}AS-47\textbf{6}	& C2 	& \textbf{12}:25 		 \\
  \textbf{F}IN-\textbf{8}0	&  C1 	& \textbf{19}:55 	 	\\
\bottomrule
	\end{tabular}};
	
\end{tikzpicture}
\quad
\begin{tikzpicture}[gate/.style={inner sep=1mm,draw,rounded corners,rectangle},scale=.75]

		\node (f0) at (1,2) {F7};
		\node (fg) at (3,1.5) {C1};
		\node (ft0) at (5,2) {09};

		\node (f1) at (1,1) {F8};
		\node (ft1) at (5,1) {19};

		\node (s1) at (1,0.25) {S6};
		\node (sg1) at (5,0.25) {C2};
		\node (st) at (3,-0.35) {12};

		\node (s0) at (1,-1) {S5};
		\node (sg0) at (5,-1) {C3};

		\foreach \f/\t in {f0/f1, f0/fg, f0/ft0, f1/fg, f1/ft1,  ft1/fg, ft0/fg, s0/s1, s0/sg0, s0/st, s1/st,  st/sg1, st/sg0, s1/sg1}{
			\draw[-] (\f) -- (\t);
		}
	\end{tikzpicture}
\qquad
\begin{tikzpicture}[level distance= 2.3 em, sibling distance= 5 em,
	every node/.style={draw, scale=0.85},
	edge from parent/.style={thin,-,black, draw}]
\node  {{F7,F8,C1}} 
		child {node {F7,C1,9}}
		child {node {F8,C1,19}
			child {node{S5,S6,12} 
			child {node {S6,12,C2}}
			child {node {S5,12,C3}}}};
\end{tikzpicture}	

\caption{An $\FO$-structure $\mathcal A= (A,S^\mathcal A,R^\mathcal A)$ (Left) with the Gaifman graph $G_{\calA}$ (Middle) and a possible treedecomposition of $G_{\calA}$ (Right) of Example~\ref{tw-ex}. For brevity, universe elements are written in short forms.}\label{fig:tw-ex}

\end{figure}
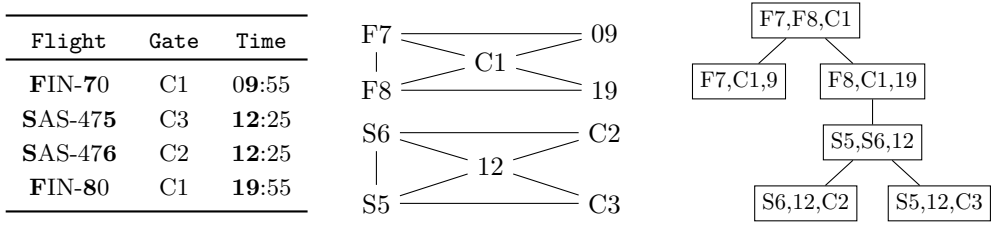
	Consider the database form our previous example.
	Recall that the universe $A$ consists of entries in each row.
	 Let $\tau = \{\mathrm{S}^2, \mathrm{R}^3\}$ include a binary relation $\mathrm{S}$ ($\mathrm{S}(x,y):$ flights $x$ and $y$ are owed by the same company) and a ternary relation $\mathrm{R}$ ($\mathrm{R}(x,y,z): $ the gate $x$ is reserved by the flight $y$ at time $z$).
	 For simplicity, we only consider first four rows with the corresponding three columns from Table~\ref{tbl:ex}, see Figure~\ref{fig:tw-ex} for an explanation.
	 Since the largest bag size in our decomposition is $3$, the treewidth of this decomposition is $2$.
	 Furthermore, the presence of cycles of length $3$ suggests that there is no better decomposition. 
	 As a consequence the given structure has treewidth $2$.
\end{example}
The decision problem to determine whether the treewidth of a given graph $\mathcal G =(V,E)$ is at most $k$, is $\NP$-complete \cite{acp87}.
See Bodlaender's Guide \cite{DBLP:journals/actaC/Bodlaender93} for an overview of algorithms that compute tree decompositions.
When considering the parameter treewidth, one usually assumes it as a given value and does not need to compute it.
We consider only the model checking problem ($\MC$) and two variants in this paper. 
First, let us define the most general version. 
\problemdef{$\combined$ (combined complexity of model checking)}{a structure $\calA$, team $T$ and a $\D$-formula $\Phi$}{$(\calA,T)\models\Phi$}

We further consider the following two variants of the model checking problem.
\problemdef{$\data$ (data complexity of model checking, $\Phi$ is fixed)}{a structure $\calA$, team $T$}{$(\calA,T)\models\Phi$}
\problemdef{$\expression$ (expression complexity of model checking, $\calA, T$ are fixed)}{a $\D$-formula $\Phi$}{$(\calA,T)\models\Phi$}

\subparagraph*{List of Parameterizations.}
Now let us turn to the parameters that are under investigation in this paper. 
We study the model checking problem of $\D$ under nine various parameters that naturally occur in an $\MC$-instance.
Let $\langle \calA, T, \Phi \rangle$ be an instance of $\MC$, where $\Phi$ is a $\D$-formula, $\calA$ is a structure and $T$ is a team over $\calA$. 
The parameter $\splits$ denotes the number of occurrences of the split operator ($\lor$), $\quantifier$ is the number of universal quantifiers in $\Phi$.
Moreover, $\variable$ (resp., $\freevar$) denotes the total number of (free) variables in $\Phi$.
The parameter $\formula$ is the size of the input formula $\Phi$, and similarly the two other size parameters are $|\calA|$ and $|T|$.
The treewidth of the structure $\calA$ (see Def.~\ref{def-tw}) is defined as the treewidth of $G_{\calA}$ and denoted by $\twstruc$.
Note that for formulas using the dependence atom $\depa{x}{y}$, one can translate to a formula using only dependence atoms where $|\mathbf y|=1$ (via conjunctions). 
That is why the arity of a dependence atom $\depa{x}{y}$ is defined as $|\mathbf x|$ and $\arity$ is the maximum arity of any dependence atom in $\Phi$. 

Let $k$ be any parameterization and $P\in \{\data, \expression, \combined\}$, then by $k$-$P$ we denote the problem $P$ when parameterized by $k$.
If more than one parameterization is considered, then we use `$+$' as a separator and write these parameters in brackets, e.g., $(\formula+\freevar)\text-\data$ as the problem $\data$ with parameterization $\formula+\freevar$. 
Finally, notice that since the formula $\Phi$ is fixed for $\data$ this implies that $\paraMC{\data}{\formula}$ is nothing but $\data$.
That is, bounding the parameter does not make sense for $\data$ as the problem $\data$ remains $\NP$-complete. 
\section{Complexity results}\label{sec:results}
\begin{table}
	\centering
	\begin{tabular}{ccccc} \toprule
	Parameter 	& $\combined$	&  $\data$  & $\expression$  	\\\midrule 
	\splits 	& $\para\PSPACE\hard^{L\ref{ec-splits}}$
				& $\para\NP^{L\ref{dc-many}}$	
				& $\para\PSPACE\hard^{L\ref{ec-splits}}$		\\
	\formula 	& $\para\NP^{L\ref{cc-formula}} $
				& $\para\NP^{R\ref{rem:dc-formulasize}}$		
				&	$\FPT^{\ref{ec-formula}}$		\\
	\freevar	& $\para\NEXP^{L\ref{cc-many}}$	
				& $\para\NP ^{L\ref{dc-many}}$
				&	$\para\NEXP^{L\ref{cc-many}}$	\\
	\twstruc	& 	$\para\NEXP ^{L\ref{cc-many}}$	 
				& 	$\para\NP	^{P\ref{dc-all}}$		
				&	$\para\NEXP ^{L\ref{cc-many}}$				\\
	\strucsize	& 	$\para\NEXP ^{L\ref{cc-many}}$	
				& 	$\FPT ^{L\ref{dc-strucsize}}$	 
				& 	$\para\NEXP ^{L\ref{cc-many}}$		\\
	\teamsize	& 	$\para\NEXP ^{L\ref{cc-many}}$
				& 	$\para\NP^{L\ref{dc-teamsize}}$	
				&	$\para\NEXP ^{L\ref{cc-many}}$				\\
	\quantifier	& 	$ \para\NP\hard ^{L\ref{cc-universal}}$  
				& 	$\para\NP	^{L\ref{dc-many}}$
				&	$\para\NP^{L\ref{ec-universal}}$	\\
	\arity		& 	$\para\PSPACE\hard^{L \ref{cc-arity}}$
				& 	$\para\NP	^{L\ref{dc-many}}$
				&	$\para\PSPACE^{L\ref{ec-arity}}$	\\				
	\variable	& 	$\para\NP ^{L\ref{cc-variables}} $   		
				& 	$\para\NP	^{L\ref{dc-many}}$	
				&	$\FPT^{L\ref{ec-variables}}$	\\
			\bottomrule			

	\end{tabular}
	\caption{Complexity classification overview. A suffix $\hard$ represents the hardness result, whereas other results are completeness. The numbers in the exponent point to the corresponding result ($Lx$ means Lemma $x$, $Px$ means Proposition $x$, $Rx$ means Remark $x$). 
	Fig.~\ref{fig:cc-pic} on page~\pageref{fig:cc-pic} is a graphical presentation of this table with a different angle.}\label{tbl:overview}
\end{table}
We begin by proving relationships between various parameterizations.
\begin{lemma} \label{para-reductions} 
The following relations among parameters hold. 
\begin{enumerate}
	\item $\formula \geq k$ for any $k \in \setdefination{\splits, \quantifier, \arity, \freevar, \variable }$, 
	\item $\strucsize 	\geq  \twstruc$. Moreover, for $\data$, $\strucsize^{O(1)} 	 \geq \teamsize$,
	\item For $\expression$, $\freevar $ is constant.
	\end{enumerate}
\end{lemma}
\begin{proof}
\begin{enumerate}
	\item Clearly, the size of the formula limits all parts of it including the parameters mentioned in the list.
	\item Notice that for data complexity, the formula $\Phi$ and consequently the number of free variables in $\Phi$ is fixed.
	Moreover, due to locality of $\D$ it holds that $T\subseteq A^{r}$, where $r$ is the number of free variables in $\Phi$. 
	That is, the team $T$ can be considered only over the free variables of $\Phi$.
	This implies that teamsize is polynomially bounded by the universe size, as $|T| \leq |\calA|^{r }$. 
	Finally, the result for $\twstruc$ follows due to Definition~\ref{def-tw}.
	This is due to the reason that in the worst case all universe elements belong to one bag in the decomposition and $\twstruc = |\calA|-1$.
	\item Notice that the team $T$ is fixed in $\expression$.
	Together with the locality of $\D$-formulas (see Def.~\ref{def-semantics}), this implies that the domain of $T$ (which is same as the set of free variables in the formula $\Phi$) is also fixed and as a result, of constant size.

\end{enumerate}

\end{proof}
\begin{remark}
	If the number of free variables ($\freevar$) in a formula $\Phi$ is bounded then the total number of variables ($\variable$) is not necessarily bounded, on the other hand, bounding $\variable$ also bounds $\freevar$.
\end{remark}

\subsection{Data complexity ($\data$)}
Classically, the data complexity of model checking for a fixed $\D$-formula $\Phi$ is $\NP$-complete~\cite{DBLP:books/daglib/0030191}.
\begin{proposition}\label{dc-all}
	For a fixed formula, the problem whether an input structure $\calA$ and a team $T$ satisfies the formula is $\NP$-complete. 
	That is, the data complexity of dependence logic is $\NP$-complete.
\end{proposition}

In this section we prove that none of the considered parameter lowers this complexity, except $|\calA|$.
The proof relies on the fact that the complexity of model checking for already a very simple formula (see below) is $\NP$-complete.

\begin{lemma}\label{dc-many}
 Let $k\in\{\splits, \freevar, \variable, \quantifier, \arity, \twstruc\}$. Then the problem $\paraMC{\data}{k}$, is $\para\NP\complete$. 
\end{lemma}
\begin{proof}
	The upper bound follows from Proposition~\ref{dc-all}.
	Kontinen~\cite[Theorem~4.9]{jarmoKontinen13} proves that the data complexity for a fixed $\D$-formula of the form $\depas{x}{y}\lor \depas{u}{v}\lor \depas{u}{v}$ is already $\NP$-complete.
	For clearity, we briefly sketch the reduction presented by Kontinen~\cite{jarmoKontinen13}.
	Let $\phi = \bigwedge\limits_{i\leq m}(\ell_{i,1}\lor \ell_{i,2}\lor \ell_{i,3})$ be an instance of $3\text-\SAT$.
	Consider the structure $\calA$ over the empty vocabulary, that is, $\tau = \emptyset$.
	Let $A=  \mathrm{Var}(\phi) \cup \{0,1,\ldots,m\}$.  
	The team $T$ is constructed over variables $\{x,y,u,v\}$ that take values from $A$.
	As an example, the clause $(p_{1}\lor \neg p_{2} \lor \neg p_{3})$ gives rise to assignments in Table~\ref{tbl:ex}.
\begin{table}[t]
	\centering
\begin{tabular}{cccc}\toprule
	$x =$ `variable' & $y =$ `parity' &  $u=$ `clause' & $v=$ `position' \\ \midrule
	$p_{1}$		& 	$1$	&	$1$	&	$0$ \\
	$p_{2}$		& 	$0$	&	$1$	&	$1$ \\
	$p_{3}$		& 	$0$	&	$1$	&	$2$ \\
\bottomrule
\end{tabular}
	\caption{An example team for~$(p_{1}\lor \neg p_{2} \lor \neg p_{3})$}\label{tbl:ex}
\end{table}	
	Notice that, a truth assignment $\theta$ for $\phi$ is constructed using the division of $T$ according to each split. 
	That is, $T\models \depas{x}{y}\lor \depas{u}{v}\lor \depas{u}{v}$ if and only if $\exists P_0,P_1,P_2$ such that $\cup_i P_i= T$ for $i\leq 2$ and each $P_i$ satisfies $i$th dependence atom. 
	Let $P_0$ be such that $P_0\models \depas{x}{y}$, then we let 
	$\theta(p_{j})=1 \iff \exists s\in P, \text{ s.t. } s(x)= p_{j}$ and $s(y)=1$.
	That is, one literal in each clause must be chosen in such a way that satisfies this clause, whereas, the remaining two literals per each clause are allowed to take values that does not satisfy it.
	As a consequence, each clause is satisfied by the variables chosen in this way, which proves correctness. 
	
	This implies that the $2$-slice (for $\paraMC{\data}{\splits}$), $4$-slice (for $\paraMC{\data}{\freevar}$ as well as $\paraMC{\data}{\variable}$),  $0$-slice (for $\paraMC{\data}{\quantifier}$), and $1$-slice (for $\paraMC{\data}{\arity}$) are $\NP$-complete.
	Consequently, the $\para\NP$-hardness for these cases follow.
	Finally, the case for $\twstruc$ also follows due to the reason that the vocabulary of the reduced structure is empty. 
	As a consequence, our definition~\ref{def-tw} yields a tree decomposition of width $1$ trivially as no elements of the universe are related.

	This completes the proof to our lemma.
\end{proof}
\begin{remark}\label{rem:dc-formulasize}
Recall that $\formula$ as a parameter for $\data$ does not make sense as the input consists of $\langle \calA, T\rangle $.
That is, the formula $\Phi$ is already fixed which is stronger than fixing the size of $\Phi$.
\end{remark}
We now prove the only tractable case for the data complexity.

\begin{lemma}\label{dc-strucsize}
	$\paraMC{\data}{\strucsize}\in\FPT$.
\end{lemma}
\begin{proof}
	Notice first that restricting the universe size $|\calA|$ polynomially bounds the teamsize $|T|$, due to Lemma~\ref{para-reductions}.
		This implies that the size of whole input is (polynomially) bounded by the parameter $|\calA|$.
		The result follows trivially because any PP $P$ is $\FPT$ when the input size is bounded by the parameter \cite{DBLP:series/txtcs/FlumG06}.
\end{proof}

\begin{lemma}\label{dc-teamsize}
	$\paraMC{\data}{\teamsize}$ is $\para\NP$-complete.
\end{lemma}
\begin{proof}
	For a fixed sentence $\Phi \in \D$ (that is, with no free variables) and for all models $\calA$ and team $T$ we have that $(\calA,T)\models \Phi \iff (\calA, \{\emptyset\}) \models \Phi$. 
	As a result, the problem $\fptreduction$-reduces to the model checking problem with  $\teamsize=1$. 
	Consequently, 1-slice of $\paraMC{\data}{\teamsize}$ is $\NP$-complete because model checking for a fixed $\D$-sentence is also $\NP$-complete \cite{DBLP:books/daglib/0030191}. 
	This gives $\para\NP$-hardness.
	
	For the membership, note that given a structure $\calA$ and a team $T$ then for a fixed formula $\Phi$ the question whether $(\calA,T)\models \Phi$ is in $\NP$.
	Consequently, giving $\para\NP$-membership.
\end{proof}
A comparison with the propositional dependence logic ($\PDL$) at this point might be interesting.
If the formula size is a parameter then the model checking for $\PDL$ can be solved in $\FPT$-time~\cite{DBLP:conf/foiks/0002M20}.
However, this is not the case for $\D$ even if the formula is fixed in advance.
\subsection{Expression and Combined Complexity ($\expression, \combined$)}
Now we turn towards the expression and combined complexity of model checking for $\D$.
Here again, in most cases the problem is still intractable for the combined complexity.
However, expression complexity when parameterized by the  formula size ($\formula$) and the total number of variables ($\variable$) yields membership in $\FPT$.
Similar to the previous section, we first present results that directly translate from the known reductions for proving the $\NEXP$-completeness for $\D$.

\begin{lemma}\label{cc-many}
	Let $k \in \setdefination{\strucsize, \twstruc, \teamsize, \freevar}$. 
	Then both $\paraMC{\combined}{k}$ and $\paraMC{\expression}{k}$ are $\para\NEXP$-complete. 
\end{lemma}
\begin{proof}
	In the classical setting, $\NEXP$-completeness of the expression and the combined complexity for $\D$ was shown by Grädel~\cite[Theorem~5.1]{Gradel13}.
	This immediately gives membership in $\para\NEXP$. 
	Interestingly, the universe in the reduction consists of $\{0,1\}$ with empty vocabulary and the formula obtained is a $\D$-sentence.
	This implies that $2$-slice (for $|\calA|$), $1$-slice (for $\twstruc$), $1$-slice (for $\teamsize$), and $0$-slice (for the number of free variables) are $\NEXP$-complete.
	As a consequence, $\para\NEXP$-hardness for the mentioned cases follows and this completes the proof.
\end{proof}
For the number of splits as a parameterization, we only know that this is also highly intractable, with the precise complexity open for now.
\begin{lemma}\label{ec-splits}
	$\paraMC{\expression}{\splits}$ and $\paraMC{\combined}{\splits}$ are both $\para\PSPACE\hard$.
\end{lemma}
\begin{proof}
	Consider the equivalence of $\{\exists,\forall,\land\}\text-\FO\text-\MC$ to quantified constraint satisfaction problem (QCSP)~\cite[p.~418]{DBLP:conf/cie/Martin08}.
	That is, the fragment of $\FO$ with only operations in $ \{\exists,\forall,\land \}$ allowed.
	Then QCSP asks, whether the conjunction of quantified constraints ($\FO$-relations) is true in a fixed $\FO$-structure $\calA$.
	This implies that already in the absence of a split operator (even when there are no dependence atoms), the model checking problem is $\PSPACE$-hard.
	Consequently, the mentioned results follow.
\end{proof}

The formula size as a parameter presents varying behaviour depending upon if we consider the expression or the combined complexity.
\begin{lemma}\label{cc-formula}
	$\paraMC{\combined}{\formula}$ is \para\NP-complete.
\end{lemma}
\begin{proof}
	Notice that, due to Lemma~\ref{para-reductions}, the size $k$ of a formula $\Phi$ also bounds the maximum number of free variables in any subformula of $\Phi$.
	This gives the membership in conjunction with \cite[Theorem~5.1]{Gradel13}.
	That is, the combined complexity of $\D$ is $\NP$-complete if maximum number of free variables in any subformuala of $\Phi$ is fixed.
	The lower bound follows because of the construction by Kontinen~\cite{jarmoKontinen13} (see also Lemma~\ref{dc-many}) since for a fixed formula (of fixed size), the problem is already $\NP$-complete.
\end{proof}

\begin{lemma}\label{ec-formula}
	$\paraMC{\expression}{\formula}$ is in \FPT.
\end{lemma}
\begin{proof}
	Recall that in expression complexity, the team $T$ and the structure $\calA$ are fixed. 
	Whereas, the size of the input formula $\Phi$ is a parameter. 
	The result follows trivially because any PP $P$ is $\FPT$ when the input size is bounded by the parameter.
\end{proof}
The expression complexity regarding the number of universal quantifiers as a parameter drops down to $\para\NP$-completeness, which is still intractable but much lower than $\para\NEXP$-completeness. 
However, regarding the combined complexity we can only prove the membership in $\XNP$, with $\para\NP$-lower bound.

\begin{lemma}\label{ec-universal}
	$\paraMC{\expression}{\quantifier}$ is $\para\NP$-complete.
\end{lemma}
\begin{proof}
	We first prove the lower bound through a reduction form the satisfiability problem for propositional dependence logic ($\PDL$).
	That is, given a $\PDL$-formula $\phi$, whether there is a team $T$ such that $T\models\phi$?
	Let $\phi$ be a $\PDL$-formula over propositional variables $p_1,\ldots, p_n$.
	For $i\leq n$, let $x_i$ denote a variable corresponding to the proposition $p_i$.
	Let $\calA=\{0,1\}$ be the structure over empty vocabulary.
	Clearly $\phi$ is satisfiable iff $\exists p_1\ldots \exists p_n \phi$ is satisfiable iff $(\calA, \{\emptyset\}) \models \exists x_1\ldots \exists x_n \phi'$, where $\phi'$ is a $\D$-formula obtained from $\phi$ by simply replacing each proposition $p_i$ by the variable $x_i$.
	Notice that the reduced formula does not have any universal quantifier, that is $\quantifier(\phi')=0$.
	This gives $\para\NP$-hardness since the satisfiability for $\PDL$ is $\NP$-complete~\cite{LohmannV13}.
	
	For membership, notice that a $\D$-sentence $\Phi$ with $k$ universal quantifiers can be reduced in $\Ptime$-time to an $\ESO$-sentence $\Psi$ of the form $\exists f_1\ldots \exists f_r \forall x_1 \ldots\forall x_k \psi$ \cite[Cor.~3.9]{durand11}, where $\psi$ is a quantifier free $\FO$-formula, $r\in\mathbb N$, and each function symbol $f_i$ is at most $k$-ary for $1\leq i\leq r$.
	Finally, $(\mathcal{A},\{\emptyset\}) \models \Phi \iff \mathcal{A} \models \bigvee \limits_{f_1} \ldots \bigvee \limits_{f_r} \forall x_1 \ldots\forall x_k \psi'$. 
	Where the latter question can be solved by guessing an interpretation for each function symbol $f_i$ and $i\leq r$.
	This requires $r\cdot|\mathcal A|^k$ guessing steps, and can be achieved in $\para\NP$-time for a fixed structure $\mathcal{A}$ (as we consider expression complexity).
	Consequently, the membership in $\para\NP$ follows.
	Notice that the arity of function symbols in the $\para\NP$-membership above is bounded by $k$ if $\Phi$ is a $\D$-sentence.
	However, if  $\Phi$ is a $\D$-formulas with $m$ free variables then the arity of function symbols as well as the number of universal quantifiers in the reduction, both are bounded by $k+m$ where $k=\quantifier(\Phi)$ and $m =\freevar(\Phi)$.
	Nevertheless, recall that for $\expression$, the team is also fixed. 
	Moreover, due to Lemma~\ref{para-reductions} the collection of free variables in $\Phi$ has constant size.
	This implies that the reduction above provides an $\ESO$-sentence with $k+m$ universal quantifiers as well as  function symbols of arity $k+m$ at most.
	Finally, guessing the interpretation for functions still takes $\para\NP$-steps (because $m$ is constant) and consequently, we get $\para\NP$-membership for open formulas as well.
\end{proof}

	The following corollary immediately follows from the proof above.
	\begin{corollary}\label{ec-universal-freevar}
		$\paraMC{\expression}{(\quantifier+\freevar)}$ is $\para\NP$-complete.
	\end{corollary}

\begin{lemma}\label{cc-universal}
	$\paraMC{\combined}{\quantifier}$ is $\para\NP$-hard. Moreover, for sentences of $\D$, $\paraMC{\combined}{\quantifier}$ is in~$\XNP$.
\end{lemma}
\begin{proof}
	The \para\NP-lower bound follows due to the fact that the expression complexity of $\D$ is already $\para\NP$-complete when parameterized by $\quantifier$ (Lemma~\ref{ec-universal}).  

	For sentences, similar to the proof in Lemma~\ref{ec-universal}, a $\D$-sentence $\Phi$ can be translated to an equivalent $\ESO$-sentence $\Psi$ in polynomial time.
	However, if the structure is not fixed as for expression complexity, then the computation of interpretations for functions can no longer be done in $\para\NP$-time, but requires non-deterministic $|\calA|^{k}$-time for each guessed function, where $k=\quantifier$.	
	Consequently, we reach only membership in $\XNP$ for sentences.
\end{proof}
For open formulas, we do not know if $\paraMC{\combined}{\quantifier}$ is also in $\XNP$.
Our proof technique does not immediately settle this case as the team is not fixed for $\combined$.

Similar to the case of universal quantifiers, the arity as a parameter also reduces the complexity but not as much as the universal quantifiers. 
Moreover, the precise combined complexity when parameterized by the arity is also open.
\begin{lemma}\label{ec-arity}
	$\paraMC{\expression}{\arity}$ is $\para\PSPACE$-complete.
\end{lemma}
\begin{proof}
	Notice that a $\D$-sentence $\Phi$ with $k$-ary dependence atoms can be reduced in $\Ptime$-time to an $\ESO$-sentence $\Psi$ of the form $\exists f_1\ldots \exists f_r \psi$~\cite[Thm.~3.3]{durand11}, where $\psi$ is an $\FO$-formula and each function symbol $f_i$ is at most $k$-ary for $1\leq i\leq r$.
	Finally, $\mathcal{A} \models \Phi \iff \mathcal{A} \models \bigvee \limits_{f_1} \ldots \bigvee \limits_{f_r} \psi'$.
	That is, one needs to guess the interpretation for each function symbol $f_i$, which can be done in $\para\NP$-time.
	Finally, evaluating an $\FO$-formula $\psi'$ for a fixed structure $\calA$ can be done in $\PSPACE$-time.
	This yields membership in $\para\PSPACE$.
	Moreover, if $\Phi$ is an open $\D$-formula then the result follows due to a similar discussion as in the prof of Lemma~\ref{ec-universal}.
	
	For hardness, notice that the expression complexity of $\FO$ is $\PSPACE$-complete.
	This implies that already in the absence of any dependence atoms, the complexity remains $\PSPACE$-hard, as a consequence, the $0$-slice of $\paraMC{\expression}{\arity}$ is $\PSPACE$-hard.
	
	This proves the desired result.
\end{proof}
The combination ($\arity + \freevar $) also does not lower the expression complexity as discussed before in the case of $\quantifier$.

	\begin{corollary}\label{ec-arity-freevar}
		$\paraMC{\expression}{(\arity+\freevar)}$ is $\para\PSPACE$-complete.
	\end{corollary}

\begin{lemma}\label{cc-arity}
	$\paraMC{\combined}{\arity}$ is $\para\PSPACE$-hard.
\end{lemma}
\begin{proof}
	Consider the fragment of $\D$ with only dependence atoms of the form $\depas{}{x}$, the so-called constancy logic.
	The combined complexity of constancy logic is $\PSPACE$-complete~\cite[Theorem~5.3]{Gradel13}.
	This implies that the $0$-slice of $\paraMC{\combined}{\arity}$ is $\PSPACE$-hard, proving the result.
\end{proof}
The combined complexity of model checking for constancy logic is $\PSPACE$~\cite[Thm.~5.3]{Gradel13}.
Aiming for an $\para\PSPACE$-upper bound via squeezing the fixed arity of dependence atoms (in some way) into constancy atoms is unlikely to happen as $\D$ captures $\ESO$ whereas constancy logic for sentences (and also open formulas) collapses to  $\FO$~\cite{Galliani16}. 

Notice that a similar reduction as in the proof of Lemma~\ref{ec-universal} holds from $\PL$, in which both parameters ($\quantifier$ and $ \arity$) are bounded. 
This implies that there is no hope for tractability even when both parameters are considered together.
That is, the complexity of expression complexity remains $\para\NP$-complete when parameterized by the combination of parameters (\quantifier, \arity).
\begin{corollary}\label{cor:forall+arity}
	$\paraMC{\expression}{(\quantifier+\arity)}$ is also $\para\NP$-complete.
\end{corollary}
Finally, for the parameter total number of variables, the expression complexity drops to $\FPT$ whereas, the combined complexity drops to $\para\NP$-completeness.
The case of expression complexity is particularly interesting. 
This is due to the reason that it was posed as an open question in \cite{jonniphd} whether the expression complexity of the fixed variable fragment of dependence logic ($\D^k$) is $\NP$-complete similar to the case of the combined complexity therein.
We answer this negatively by stating $\FPT$-membership for $\paraMC{\expression}{\variable}$, which as a corollary proves that the expression complexity of $\D^k$ is in $\Ptime$ for each $k\geq 1$.
\begin{lemma}\label{cc-variables}
	 $\paraMC{\combined}{\variable}$ is $\para\NP$-complete.
\end{lemma}
\begin{proof}
Notice that if the total number of variables in $\Phi$ is fixed, then the number of free variables in any subformula $\psi$ of $\Phi$ is also fixed.
This implies the membership in $\para\NP$ due to \cite[Theorem~5.1]{Gradel13}.
On the other hand, by \cite[Theorem~3.9.6]{jonniphd} we know that the combined complexity of $\logicClFont{D}^k$ is $\NP$-complete. 
This implies that for each $k$, the $k$-slice of the problem is $\NP$-hard.
This gives the desired lower bound.
\end{proof}

\begin{lemma}\label{ec-variables}
	 $\paraMC{\expression}{\variable}$ is $\FPT$.
\end{lemma}
\begin{proof}
Given a formula $\Phi$ of dependence logic with $k$ variables, we can construct an equivalent formula $\Psi$ of $\ESO^{k+1}$ in polynomial time \cite[Theorem 3.3.17]{jonniphd}.
Moreover, since the structure $\mathcal{A}$ is fixed, there exists a reduction of $\Psi$ to an $\FO$-formula $\psi$ with $k+1$ variables (big disjunction on the universe elements for each second order existential quantifier).
Finally, the model checking for $\FO$-formulas with $k$ variables is solvable in time $O(|\psi|\cdot |A|^k)$ \cite[Prop~6.6]{DBLP:books/sp/Libkin04}. 
This implies the membership in \FPT.
\end{proof}

\begin{corollary}
	The expression complexity of $\D^k$ is in $\Ptime$ for every $k\geq 1$.
\end{corollary}
\begin{proof}
	Since both, the number of variables and the universe size is fixed. 
	The runtime of the form $O(|\psi|\cdot |A|^k)$ in Lemma~\ref{ec-variables} implies membership in $\Ptime$.
\end{proof}

\section{Conclusion}\label{sec:concl}
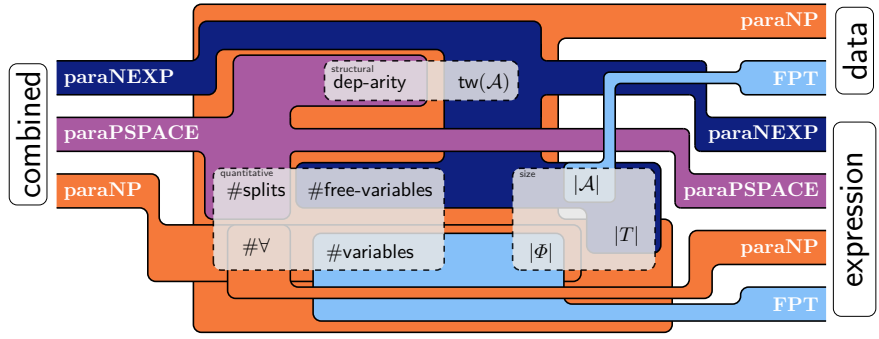
\begin{figure}
	\centering\scalebox{.75}{
	\begin{tikzpicture}
	
		\node[rotate=90,draw=black,rounded corners,text width=1.3cm,align=center] at (10.5,3.48) {\sffamily\LARGE data};
		\draw[rounded corners,,thick,black,fill=paraNP] (10,4.3) -- (-1.1,4.3) -- (-1.1,-1.5) -- (7.3,-1.5) -- (7.3,.5) -- (5.3,.5) -- (5.3,3.7) -- (10,3.7);

		\node[anchor=east] at (10,4) {\color{white}$\para\NP$};
		\node[anchor=east] at (10,3) {\color{white}$\FPT$};

		\node[rotate=90,draw=black,rounded corners] at (-4,2.) {\sffamily\LARGE combined};
		\draw[rounded corners,,thick,black,fill=paraNEXP] (-3.5,3.3) -- (-1,3.3) -- (-1,4) -- (5,4) -- (5,1.5) -- (7.1,1.5) -- (7.1,-.1) -- (5.8,-.1) -- (5.8,.7) -- (.7,.7) -- (.7,1.5) -- (3.3,1.5) -- (3.3,3.5) -- (-.7,3.5) -- (-.7,2.7) -- (-3.5,2.7);
		\node[anchor=west] at (-3.5,3) {\color{white}$\para\NEXP$};
		\draw[rounded corners,,thick,black,fill=paraPSPACE] (-3.5,2.3) -- (-.4,2.3) -- (-.4,3.4) -- (3,3.4) -- (3,2.5) -- (.6,2.5) -- (.6,0.5) -- (-0.9,.5) -- (-.9,1.7) -- (-3.5,1.7);
		\node[anchor=west] at (-3.5,2) {\color{white}$\para\PSPACE$};
		\draw[rounded corners,,thick,black,fill=paraNP] (-3.5,1.3) -- (-1.7,1.3) -- (-1.7,.4) -- (5.7,.4) -- (5.7,-.6) -- (-1.9,-.6) -- (-1.9,.7) -- (-3.5,.7); 
		\node[anchor=west] at (-3.5,1) {\color{white}$\para\NP$};		
	
		\node[rotate=90,draw=black,rounded corners,text width=3.2cm,align=center] at (10.5,.5) {\sffamily\LARGE expression};

		\path[fill=paraNEXP,apply={draw=black,thick} except on segments {8,9,10}] (10,2.3) {[rounded corners=1mm] -- (7.9,2.3) -- (7.9,3.3) -- (5,3.3)}-- (5,3.6) -- (4.7,3.6) -- (4.7,2.3) -- (5,2.3){[rounded corners=1mm] -- (5,2.7) -- (7.7,2.7) -- (7.7,1.7) -- (10,1.7)};

		\path[fill=paraPSPACE,apply={draw=black,thick} except on segments {8,9,10}] (10,1.3) {[rounded corners=1mm] -- (7.6,1.3) -- (7.6,2.1)-- (0.6,2.1)} -- (0.6,2.3) -- (0.3,2.3) -- (0.3,1.5) -- (.6,1.5){[rounded corners=1mm]  -- (.6,1.7) -- (7.35,1.7) -- (7.35,.7) -- (10,.7)};

		\draw[rounded corners,,thick,black,fill=FPT] (10,-.7) -- (8.5,-.7) -- (8.5,-1) -- (5.4,-1) -- (5.4,0.25) -- (1,.25) -- (1,-1.3) -- (8.8,-1.3) -- (10,-1.3);
		\draw[rounded corners,,thick,black,fill=paraNP] (10,.3) -- (7.6,.3) -- (7.6,-.7) -- (0.6,-.7) -- (.6,.4) -- (-.5,.4) -- (-.5,-.9) -- (8,-.9) -- (8,-.3) -- (10,-.3);
		\node[anchor=east] at (10,2) {\color{white}$\para\NEXP$};		
		\node[anchor=east] at (10,1) {\color{white}$\para\PSPACE$};
		\node[anchor=east] at (10,0) {\color{white}$\para\NP$};
		\node[anchor=east] at (10,-1) {\color{white}$\FPT$};

		\draw[rounded corners=1mm,thick,black,fill=FPT] (10,3.3) -- (8.5,3.3) -- (8.5,3.1) -- (6.1,3.1) -- (6.1,1.5) -- (5.4,1.5) -- (5.4,0.8) -- (6.3,.8) -- (6.3,2.9) -- (8.5,2.9) -- (8.5,2.7) -- (10,2.7);
		\node[anchor=east] at (10,3) {\color{white}$\FPT$};

		\begin{scope}[yshift=-1mm]
			\draw[opacity=.85,rounded corners,dashed,thick,black,fill=gray!20!white] (1.2,2.7) rectangle (4.6,3.4);
			\node[anchor=west] at (1.2,3.27) [opacity=.8] {\tiny\sffamily structural};
			\node at (4,3) {$\twstruc$};		
			\node at (2,3) {$\arity$};	
		\end{scope}
		\draw[opacity=.85,rounded corners,dashed,thick,black,fill=gray!20!white] (-.75,-.4) rectangle (3.3,1.4);
		\node[anchor=west] at (-.75,1.3) [opacity=.8] {\tiny\sffamily quantitative};
		\node at (2,1) {$\freevar$};
		\node at (0,1) {$\splits$};
		\node at (0,0) {$\quantifier$};
		\node at (2,-.1) {$\variable$};
		
		\draw[opacity=.85,rounded corners,dashed,thick,black,fill=gray!20!white] (4.5,-.4) -- (7,-.4) -- (7,1.4) -- (4.5,1.4) -- cycle;
		\node[anchor=west] at (4.5,1.3) [opacity=.8] {\tiny\sffamily size};
		\node at (6.5,0.2) {$\teamsize$};
		\node at (5.8,1.1) {$\strucsize$};
		\node at (5,-.1) {$\formula$};

	\end{tikzpicture}}
	\caption{Complexity classification overview for model checking problem of dependence logic, that takes grouping of parameters (quantitative, size, structural) and complexity classes into account.}\label{fig:cc-pic}
\end{figure}

In this paper, we started the parameterized complexity classification of model checking for dependence logic $\D$ with respect to nine different parameters (see Table~\ref{tbl:overview} for an overview of the results). 
In Fig.~\ref{fig:cc-pic} we depict a different kind of presentation of our results that also takes the grouping of parameters into quantitative, size related, and structural into account. 
The data complexity of $\D$ shows a dichotomy ($\FPT$ vs./ $\para\NP$-complete), where surprisingly there is only one case ($\strucsize$) where one can reach $\FPT$. 
This is even more surprising in the light of the fact that the expression ($\expression$ and the combined ($\combined$) complexities under the same parameter are still highly intractable.
Furthermore, there are parameters when $\combined$ and $\expression$ vary in the complexity ($\variable$).
The combined complexity of $\D$ stays intractable under any of the investigated parameterizations.
It might be interesting to study combination of parameters and see their joint effect on the complexity (yet, Corollaries~\ref{ec-universal-freevar},~\ref{ec-arity-freevar},~\ref{cor:forall+arity} tackle already some cases). 

We want to close this presentation with some further questions:
\begin{itemize}
	\item What other parameters could be meaningful (e.g., number of conjunction, number of existential quantifiers, treewidth of the formula)?
	\item What is the exact complexity of $\quantifier$-$\combined$, $\splits$-$\expression$/-$\combined$, $\arity$-$\combined$?
	\item The parameterized complexity analysis for other team-based logics, such as independence logic and inclusion logic.
\end{itemize}

\bibliography{main}
\end{document}